\documentclass[twoside,journal]{IEEEtran}
\usepackage{amsfonts}
\usepackage{amssymb}
\usepackage{amsmath}
\usepackage{mathrsfs}
\usepackage{verbatim}
\usepackage{subfigure}
\usepackage{balance}
\usepackage{booktabs}
\usepackage{fancybox}
\usepackage{bm}
\usepackage{changepage}
\usepackage{extarrows}
\usepackage{algorithm}
\usepackage{algorithmic}
\usepackage{diagbox}
\usepackage{multirow}
\usepackage{array}
\usepackage{graphicx}
\usepackage{epstopdf}
\usepackage{caption}
\usepackage{color}
\newtheorem{theorem}{Theorem}

\newtheorem{corollary}{Corollary}
\newtheorem{proof}{Proof}
\newtheorem{proposition}{Proposition}

\begin{document}

	\title{Secure and Energy-Efficient Transmissions in Cache-Enabled Heterogeneous Cellular Networks: Performance Analysis and Optimization}
	\author{Tong-Xing~Zheng,~\IEEEmembership{Member,~IEEE,}~Hui-Ming~Wang,~\IEEEmembership{Senior~Member,~IEEE,~}and~Jinhong~Yuan,~\IEEEmembership{Fellow,~{IEEE}}
		\thanks{T.-X.~Zheng and H.-M. Wang are with the School of Electronic and Information Engineering,
			Xi'an Jiaotong University, Xi'an 710049, China, and
			also with the Ministry of Education Key Laboratory for Intelligent Networks
			and Network Security, Xi'an Jiaotong University, Xi'an 710049, China (e-mail: zhengtx@mail.xjtu.edu.cn, xjbswhm@gmail.com).
		}
		\thanks{J. Yuan is with the School of Electrical Engineering and Telecommunications, University of New South Wales, Sydney, NSW 2052, Australia (e-mail: j.yuan@unsw.edu.au).}
	}
	
	\maketitle
	\vspace{-0.8 cm}
	
\begin{abstract}
This paper studies physical-layer security for a cache-enabled heterogeneous cellular network comprised of a macro base station and multiple small base stations (SBSs). We investigate a joint design on caching placement and file delivery for realizing secure and energy-efficient transmissions against randomly distributed eavesdroppers. 
We propose a novel hybrid ``most popular content'' and ``largest content diversity'' caching placement policy to distribute files of different popularities. 
Depending on the availability and placement of the requested file, we employ three cooperative transmission schemes, namely, distributed beamforming, frequency-domain orthogonal transmission, and best SBS relaying, respectively.
We derive analytical expressions for the connection outage probability and secrecy outage probability for each transmission scheme. 
Afterwards, we design the optimal transmission rates and caching allocation successively to achieve a maximal overall secrecy throughput and secrecy energy efficiency, respectively.
Numerical results verify the theoretical analyses and demonstrate the superiority of the proposed hybrid caching policy.
\end{abstract}
	
\begin{IEEEkeywords}
	Physical-layer security, wireless caching, heterogeneous cellular network, cooperative transmission, secrecy outage probability, energy efficiency, stochastic geometry.
\end{IEEEkeywords}
	
\IEEEpeerreviewmaketitle
	
\section{Introduction}

\IEEEPARstart{H}{uman} society is striding into the era of  Internet-of-Everything (IoE), and the amount of wireless data traffic is expected to soar roughly 1000$\times$ in the coming decade. 
This poses an unprecedented challenge to backhaul links and backhaul capacity becomes a crucial system bottleneck.
Against this background, wireless caching techniques, via which popular content can be pre-stored at the edge of a wireless network before being requested by users, have emerged as a promising approach for relieving the backhaul bottleneck.
Recently, wireless caching has been substantially proven to be highly effective in aspects of alleviating  the enormous backhaul demand
\cite{Wang2014Cache, Xiang2017Cross}, {reducing the file service latency \cite{Shanmugam2013FemtoCaching}-\cite{Friedlander2018Generalization}}, increasing the successful delivery probability \cite{Chen2017Cooperative}, and improving the energy efficiency \cite{Liu2014Cache,Liu2014Will}, etc.

Information security is a fundamental concern for a cache-enabled wireless network.
Unlike early caching techniques applied for wired networks such as the Internet, wireless caching suffers a serious security vulnerability because of the openness of physical media. 
{Using traditional cryptographic encryption solely to protect content secrecy against malicious eavesdroppers will encounter three major difficulties. First and foremost, with the rapid development of quantum computing and data analytics, the computing-complexity based cryptography is under severe threat since a confidential message could easily be decrypted by brute force when eavesdropper has an extremely high level of computing power. In the second place, the encrypted content is tailored uniquely for each user request and cannot be reused for other user requests \cite{Paschos2016Wireless}. Last but not the least, the storage, management, and distribution of secret keys are troublesome to implement in emerging wireless networks due to the increasingly dynamic and large-scale network topologies \cite{Poor2012Information}. 
All these shortcomings might outweigh the advantages of wireless caching such as high flexibility, multiplexing gains, and resource utilization efficiency.  Fortunately, 
\emph{physical-layer security} \cite{Wyner1975Wire-tap}, which aims to achieve information-theoretic  secrecy by means of channel coding and exploiting inherent characteristics of wireless media, e.g., random noise and fluctuate channels, is expected to be a powerful supplement of the encryption security mechanisms. 
A considerable body of literature has proven that physical-layer security can gain remarkable
secrecy enhancement for various wireless networks \cite{Wang2016Physical_book}-\cite{Ghosh2017Secrecy}.
In view of this, it is meaningful to explore the potential of physical-layer security for cache-enabled wireless networks.}

\subsection{Previous Works}

Existing research involving the security issue for wireless caching networks has mainly concentrated on encryption and coding based on an information-theoretic framework built in \cite{Maddah-Ali2014Fundamental}. For example, the authors in \cite{Sengupta2015Fundamental} have proposed a coded caching scheme as per Shannon's one-time pad encryption to achieve secrecy. 
However, such a coded scheme requires a large quantity of random secret keys, and secure sharing of these keys will create substantial overhead. 
This work has later been extended to a device-to-device network in \cite{Awan2015Fundamental}, where a sophisticated key generation and encryption scheme has been designed.
The authors in \cite{Gerami2015Secure} have investigated the security-wise content placement based on the maximum distance separable codes. 
In \cite{Gabry2016On}, secure caching placement has been devised to prevent eavesdropper from intercepting a sufficient number of coded packets for successfully recovering video files. {In \cite{Balasubramanian2013Secure,Li2018Fundamental}, the fundamental limits on the secure rate for both symmetric and asymmetric multilevel diversity decoding systems have been explored with different-importance messages stored distributively in several sources.}

It should be stressed that, the  predominant security menace to a wireless caching network stems from content delivery rather than content placement, since the former is more susceptible to eavesdropping attacks. 
In response to this risk, it is urgent to seek effective wireless security mechanisms, particularly taking into account the intrinsic characteristics of wireless channels. This unfortunately has not yet been reported in the aforementioned endeavors.
In two recent works \cite{Xiang2017Cache,Xiang2018Secure}, physical-layer security in a multi-cell wireless caching network has been investigated. The authors have exploited cooperative multi-antenna transmissions to increase the secrecy rate against a single eavesdropper in \cite{Xiang2017Cache} and multiple untrusted cache helpers in \cite{Xiang2018Secure}. 
Nevertheless, an overly optimistic assumption has been made therein  that the instantaneous channel state information (CSI) of the eavesdropper can be estimated, which is difficult to realize in a real-world wiretap scenario because the eavesdropper usually listens passively and even remains silent to hide its existence. Moreover, the impacts of channel fading and the uncertainty of eavesdroppers' locations on security performance have not been assessed in \cite{Xiang2017Cache,Xiang2018Secure}.  
To the best of our knowledge, the benefit of physical-layer security for a cache-enabled wireless network has not yet been excavated in the presence of randomly distributed eavesdroppers. In particular, a comprehensive performance analysis and optimization framework on physical-layer security from a secrecy outage perspective is still lacking. These motivate our research work.
 
\subsection{Our Work and Contributions}
This paper examines physical-layer security for a cache-enabled macro/small-base-station (MBS/SBS) cellular network coexisting with randomly located eavesdroppers.
We establish a design framework combining both caching placement and wireless transmission, and provide a comprehensive analysis and optimization on security performance in terms of throughput and energy efficiency.
The main contributions of this paper are summarized as follows:

\begin{itemize}
\item We put forward a novel hybrid ``most popular content (MPC)'' and ``largest content diversity (LCD)'' caching placement policy to assign different-popularity files to the SBSs. 
We then employ distributed beamforming, frequency-domain orthogonal transmission, and best SBS relaying as transmission schemes for the situations when the requested file is stored at the MPC or LCD caching mode, or is not cached by the SBSs, respectively.

\item We assess the connection outage probability (COP) and  secrecy outage probability (SOP) of the content delivery for each transmission scheme with tractable expressions provided.
We also make an analytical comparison between these schemes and summarize their pros and cons in terms of the COP and SOP, respectively.

\item
We reveal a nontrivial trade-off between content diversity, throughput, and energy efficiency under the proposed hybrid caching policy. Subsequently, we optimize the overall secrecy throughput and secrecy energy efficiency by jointly determining the transmission rates and caching allocation. We derive explicit solutions on the optimal rate and caching parameters, and develop various useful properties regarding them to guide practical designs.
\end{itemize}

\subsection{Organization and Notations}
The remainder of this paper is organized as follows. 
In Section II, we describe the system model. 
In Sections III, we analyze the COP and SOP for the DBF, FOT, and BSR schemes, respectively. 
In Sections IV and V, we jointly design the optimal transmission rates and the optimal caching allocation for maximizing the overall secrecy throughput and the secrecy energy efficiency, respectively. 
In Section VI, we conclude our work.

\emph{Notations}: $|\cdot|$, $\lceil \cdot\rceil$, $(\cdot)^{\dagger}$, $\ln(\cdot)$, $\mathbb{P}\{\cdot\}$, $\mathbb{E}_v[\cdot]$ denote the absolute value, round up, conjugate, natural logarithm, probability, and the expectation taken over a random variable $v$, respectively. 


\section{System Model}
\begin{figure}[!t]
\centering
\includegraphics[width = 3.0in]{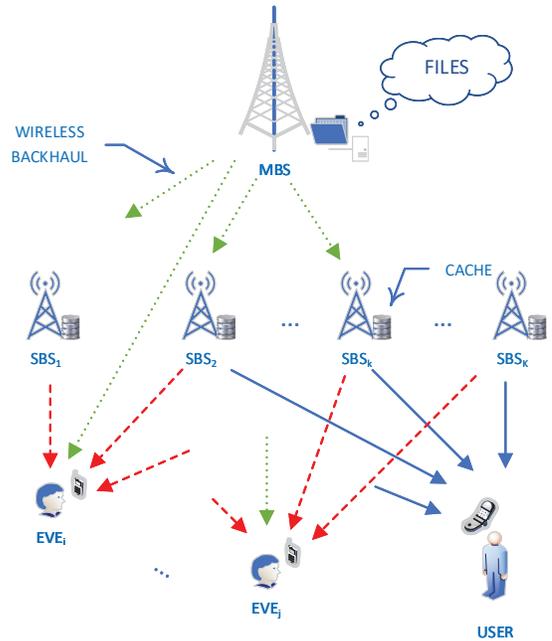}
\caption{Illustration of a cache-enabled heterogeneous cellular network. The delivery of  confidential content to a subscriber is potentially wiretapped by randomly located eavesdroppers. The wireless backhaul links from the MBS to the SBSs are insecure. }
\label{SYSTEM_MODEL}
\end{figure}
We examine a heterogeneous cellular network where an MBS and $K$ SBSs collaboratively convey secret content to a subscriber, as illustrated in Fig. \ref{SYSTEM_MODEL}. The ongoing content delivery is overheard by  potential eavesdroppers, e.g., non-paying subscribers. 
All the SBSs are connected to the MBS through wireless backhaul links. The SBSs each are equipped with a cache unit, with which they can pre-store popular content before serving local users, as a means of {shifting} the traffic burden from peak to off-peak hours. If a file requested by a subscriber is available in the SBSs, a \emph{cache hit} event is deemed to happen, and the SBSs can directly send the requested file to the subscriber; otherwise, a \emph{cache miss} event is said to have occurred, and the SBSs should first fetch the requested file from the MBS and then deliver it to the subscriber.\footnote{We assume that there is no direct transmission link from the MBS to the user due to a deep fading and a long distance.} The MBS, SBSs, subscribers, and eavesdroppers are all single-antenna devices, and only one file demand is admitted in each time slot.\footnote{{If multiple subscribers send file requests to an SBS simultaneously, in order to avoid the inter-user interference, these subscribers can be served using orthogonal multiple access methods, e.g., TDMA and FDMA.} }
Without loss of generality, we place a typical user (subscriber) at the origin of the polar coordinate. We denote the locations of the MBS, the $k$-th nearest SBS, and the $j$-th nearest eavesdropper as {$\{b: (r_{b},\theta_{b}) \}$}, $\{{s_k}:(r_{s_k},\theta_{s_k})\}$, and $\{e_j:(r_{e_j},\theta_{e_j})\}$, respectively, with $r$ and $\theta$ being the corresponding distance and direction. Since the eavesdroppers are randomly located over the network, we model their positions by a stationary Poisson point process (PPP) $\Phi_e$ on the two-dimensional plane with density $\lambda_e$, i.e., $e_j\in\Phi_e$ \cite{Zheng2015Multi}-\cite{Zheng2017Physical}. 
With a PPP model, the uncertainty of the eavesdroppers' locations can be characterized analytically using tools from the stochastic geometry theory, and key insights can be abstracted conveniently.

Wireless channels, including the main channels spanning from the SBSs to the subscriber and the wiretap channels from the MBS/SBSs to the eavesdroppers, are modeled by a frequency flat Rayleigh fading along with a standard distance-based path loss. 
Hence, the channel gain from a transmitter located at $x$ to a receiver at $y$ can be expressed as $h_{x,y}r_{x,y}^{-\alpha/2}$, where $h_{x,y}$ denotes the fading coefficient obeying the circularly symmetric complex Gaussian distribution with unit mean and zero variance, $r_{x,y}$ denotes the distance between $x$ and $y$, and $\alpha$ denotes the path-loss exponent. {We suppose that the SBSs can acquire the instantaneous CSI of their respective main channels through channel training, estimation, and feedback. We also adopt a generic hypothesis that the statistic CSI of the wiretap channels is also available at the base stations \cite{Zhou2011Throughput}-\cite{Liu2017Enhancing}.\footnote{{Typical examples include a regular user in the network, who has no authorization to access the content sent to the subscriber in a specific time slot and is thus treated as a potential eavesdropper. Although it is difficult for a base station to obtain the eavesdropper's instantaneous CSI, the base station is still capable to learn the eavesdropper' statistic CSI by collecting and analyzing an exceedingly large amount of information exchanging between the two parties during the other time slots.}} }

\subsection{Hybrid Caching Placement Policy}
The MBS possesses a library of $N$ equal-size files of different popularities. For convenience, we denote the $m$-th most popular file as $F_m$. 
{Assume that users make file demands independently as per the Zipf distribution \cite{Shanmugam2013FemtoCaching}, and then the request probability for file $F_m$ is given by
	\begin{equation}\label{zipf}
	p_{m} = \frac{m^{-\tau}}{\sum_{n=1}^{N}n^{-\tau}},
	\end{equation}
	where $\tau>0$ models the skewness of popularity distribution with a larger $\tau$ corresponding to a more concentrated popularity profile.}

We consider that the $K$ SBSs each have a caching capacity of $L$ files such that they can store up to $KL$ files. To fully utilize the limited cache resources, we propose a hybrid MPC and LCD caching placement policy to judiciously distribute files to the SBSs. 
To be specific, we first divide every file into $K$ equal-size partitions, and assign the first $M\le L$ most popular files in every SBS (i.e., MPC caching). We then use the remaining cache resources to disjointly store the $K$ partitions of the less popular files at the $K$ SBSs  (i.e., LCD caching), as a means to increase the content diversity (i.e., caching more files).
With such hybrid caching, at most $M+K(L-M)$ different files can be found in the cache units, and the files in the library can be classified into three categories: file $F_n$ with popularity order $1\le n\le M $ is cached at the MPC mode and is available in every SBS; file $F_n$ with $M < n \le M + K(L-M)$ is cached at the LCD mode, and the $K$ SBSs each hold one unique partition of $F_n$; file $F_n$ with $ n > M + K(L-M)$ is not stored by the SBSs and can only be fetched from the MBS. 
{We will later show an inherent trade-off between content diversity, reliability, and secrecy with our hybrid caching policy. An elaborative allocation between MPC and LCD caching, i.e., choosing a proper $M$, plays a critical role in balancing the above aspects and enhancing the overall security performance.}

\subsection{Cooperative Transmission Schemes}
When the $K$ SBSs receive a file request from the typical user, they should immediately deliver the requested file to it. 
Depending on the storage status of the requested file, the following three cooperative transmission schemes are employed for file delivery.

\subsubsection{Distributed Beamforming (DBF)}
If the requested file $F_n$ is cached at the MPC mode, i.e., $1\le n\le M $, the SBSs own the same copy of $F_n$. 
In order to improve transmission reliability, the SBSs jointly deliver the requested file via a distributed beamforming manner.\footnote{With DBF, each SBS only requires the local CSI of itself instead of the global CSI, which greatly lowers the system overhead.} 
Denote the weight coefficient at the $k$-th SBS as ${w_{s_k,o}} = h^{\dagger}_{s_k,o}/|h_{s_k,o}|$, and then the received signal-to-noise ratios (SNRs) at the typical user and at the $j$-th eavesdropper can be respectively expressed as
\begin{equation}\label{snr_o_db}
{\gamma}_{o} ={P_s}\left|\sum_{k=1}^K |h_{s_k,o}|r_{s_k,o}^{-\alpha/2}\right|^2,
\end{equation}
\begin{equation}\label{snr_e_db}
{\gamma}_{e_j} = {P_s}\left|\sum_{k=1}^K \frac{h^{\dagger}_{s_k,o}h_{s_k,e_j}}{|h_{s_k,o}|}r_{s_k,e_j}^{-\alpha/2}\right|^2,~~~\forall e_j\in\Phi_e,
\end{equation}
where $P_s$ denotes the SBS transmit power normalized by the receiver noise power. We consider identical noise spectral density at all the receivers.

\subsubsection{Frequency-Domain Orthogonal Transmission (FOT)}
If the requested file $F_n$ is cached at the LCD mode, i.e., $M < n \le M + K(L-M)$, the SBSs each have one different partition of $F_n$.
In order to avoid the co-channel interference, the SBSs simultaneously transmit their stored partitions in a frequency-domain orthogonal manner, i.e., disjointly using $1/K$ of the overall bandwidth. For the partition delivered from the $k$-th SBS, the received SNRs at the typical user and at the $j$-th eavesdropper can be respectively given by
\begin{equation}\label{snr_o_ot}
{\gamma}_{s_k,o} =  {KP_s}\left|h_{s_k,o}\right|^2r_{s_k,o}^{-\alpha},
\end{equation}
\begin{equation}\label{snr_e_ot}
{\gamma}_{s_k,e_j} =  {KP_s}\left|h_{s_k,e_j}\right|^2r_{s_k,e_j}^{-\alpha},~~~\forall e_j\in\Phi_e.
\end{equation}
Note that the factor $K$ exists due to the $1/K$ decrement of bandwidth available for each SBS.

\subsubsection{Best SBS Relaying (BSR)}
If the requested file $F_n$ is not cached by the SBSs, i.e., $n >M + K(L-M)$, it should be fetched from the MBS before being sent to the user. In order to balance  throughput and energy efficiency, a BSR scheme is proposed where the SBS having the highest main channel gain is chosen to forward the requested file from the MBS to the typical user. Then, the whole content delivery is divided into two hops, and there actually exists double information leakage to the eavesdroppers due to the insecure wireless backhauling. {To prevent the eavesdroppers from realizing coherent superposition of the two-hop signals, e.g., via maximal ratio combining, the selected SBS employs the decode-and-forward relaying protocol and re-encodes the message with independent codewords or even separate codebooks from those used by the MBS \cite{Zheng2015Outage}. By this means, any eavesdropper can only  demodulate the two-hop signals individually, and thus the wiretapping capability will be degraded remarkably.} Denote the index of the best SBS as $k^*$, i.e., $k^*=\arg_{1\le k\le K}\max |h_{s_k,o}|^2r_{s_k,o}^{-\alpha}$, and then the received SNRs at the typical user and at the $j$-th eavesdropper for the two hops can be respectively given by 
\begin{equation}\label{snr_o_ss}
{\gamma}_{s_{k^*},o} =  {P_s}\left|h_{s_{k^*},o}\right|^2r_{s_{k^*},o}^{-\alpha},
\end{equation}
\begin{equation}\label{snr_e_ss1}
{\gamma}_{b,e_j} =  {P_m}\left|h_{b,e_j}\right|^2r_{b,e_j}^{-\alpha},~~~\forall e_j\in\Phi_e,
\end{equation}
\begin{equation}\label{snr_e_ss2}
{\gamma}_{s_{k^*},e_j} =  {P_s}\left|h_{s_{k^*},e_j}\right|^2r_{s_{k^*},e_j}^{-\alpha},~~~\forall e_j\in\Phi_e,
\end{equation}
where $P_m$ denotes the MBS transmit power normalized by the receiver noise power. Note that wireless backhauling not only worsens transmission secrecy, but also increases the end-to-end latency as well as power consumption. All these negative impacts on the system performance will be taken into account in the subsequent analysis and optimization.

{In practice, the proposed caching policy and physical-layer transmission schemes can be implemented in two different timescales, namely, short-timescale and long-timescale, respectively. To be specific, in the short-timescale, content should be delivered on-line in each time slot based on the actual user request and the instantaneous CSI. In contrast, in the long-timescale, the cached content can be updated off-line every $T$ time slots based on the historical profiles of user preferences and the statistic CSI. Typically, we can choose $T\gg 1$, since user preferences vary on a much slower scale (e.g., in days) than user requests (e.g., in milliseconds) \cite{Xiang2018Secure}.}

\subsection{Performance Metrics}
{For the sake of secrecy, the well-known Wyner's wiretap code is employed to encode data before transmission \cite{Wyner1975Wire-tap}. 
	There are two rates in the wiretap code, namely, codeword rate $R_t$ and secrecy rate $R_s$, which are the rates of the transmitted codewords and the embedded secret messages, respectively. The rate redundancy $R_e = R_t - R_s$ reflects the cost of achieving secrecy against eavesdropping. If the main channel capacity falls below $R_t$, the intended receiver cannot recover the codeword correctly and this is regarded as connection outage. The probability that this event takes place is referred to as the COP, denoted as $\mathcal{O}_{co}$.
	If the capacity of the wiretap channel lies above $R_e$, perfect secrecy is compromised and a secrecy outage event is considered to have occurred. 
	The probability of this event happening is referred to as the SOP, denoted as $\mathcal{O}_{so}$ \cite{Zhou2011Throughput}-\cite{Zheng2017Physical}. 
	We consider non-colluding wiretapping where the eavesdroppers decode messages individually. Therefore, perfect secrecy can be ensured if confidential information is not leaked to the most deteriorate eavesdropper who has the highest wiretap channel capacity.}

In this paper, we focus on two core metrics, namely, \emph{secrecy throughput} and \emph{secrecy energy efficiency}, respectively. The two metrics measure the physical-layer security performance in terms of transmission capacity and energy efficiency from an outage point of view, respectively.

{\emph{{1)} Secrecy throughput (bits/s/Hz)}} \cite{Zhou2011Throughput}-\cite{Zheng2017Physical}, denoted as $\Psi$, is defined as the {average} successfully transmitted secret information bits per second per Hertz subject to an SOP constraint $\mathcal{O}_{so}\le \epsilon$, where $\epsilon\in[0,1]$ is a prescribed threshold SOP. 
The secrecy throughput for a specific transmission scheme, labeled as $i\in\{\rm D,F,B\}$,{\footnote{Throughout this paper, the letters ``$\rm D$'', ``$\rm F$'', and ``$\rm B$'' refer to the DBF, FOT, and BSR transmission schemes, respectively.}} can be expressed mathematically as the product of the secrecy rate and the complement of the COP, i.e., $\Psi^{i} = R^{i}_{s}(1-\mathcal{O}^{i}_{co})$.
Then, the overall secrecy throughput under the proposed caching policy can be given by
\begin{equation}\label{def_st}
\bar\Psi = \sum_{i\in\{ \rm D,F,B\}}p_t^{i} \Psi^{i},
\end{equation}
where $p_t^{i}$ denotes the probability of adopting scheme $i$, which will be detailed later for the design of caching allocation.

\emph{{2)} Secrecy energy efficiency (bits/Joule/Hz)} \cite{Zheng2017Physical}, denoted as $\Omega$, is defined as the {average} successfully transmitted classified information bits per Joule per Hertz subject to an SOP constraint. Formally, it can be expressed as the ratio of the overall secrecy throughput $\bar\Psi$ to the average consumed power $P_{\rm avg}$ for serving a user request, which is given below:
\begin{equation}\label{def_see}
\Omega =  \frac{\bar\Psi}{P_{\rm avg}}=\frac{p_t^{\rm D}\Psi^{\rm D}+p_t^{\rm F}\Psi^{\rm F}+p_t^{\rm B}\Psi^{\rm B}}{K P_s(p_t^{\rm D}+p_t^{\rm F})+p_t^{\rm B}(P_m+P_s)}.
\end{equation}

{We emphasize that the wiretap codeword rate $R_t$, secrecy rate $R_s$, rate redundancy $R_e$, and the allocation $M$ between MPC and LCD caching play significant roles in improving the secrecy throughput and secrecy energy efficiency. 
Specifically, $R_t$, $R_s$, and $R_e$ trigger a nontrivial trade-off between transmission reliability, secrecy, and throughput for each transmission scheme. Generally, high throughput requires a large $R_s$, whereas an overly large $R_s$ will inevitably increase the COP and in turn lower the throughput. Likewise, increasing $R_e$ can decrease the SOP but it will also enlarge $R_t$ thus leading to a larger COP and a lower throughput. 
Meanwhile, the allocation $M$ strikes a vital trade-off between content diversity, throughput, and energy efficiency. 
On the one hand, increasing $M$, i.e., a larger probability of adopting the DBF scheme, is profitable for throughput improvement. 
On the other hand, a too large $M$ will decrease the content diversity or increase the chance of backhauling file fetching thus introducing additional delivery delay and power consumption, which is unfortunately destructive for throughput and energy efficiency. 
The above opposite impacts on the system performance need to be weighed carefully.}

In this paper, we desire to optimize the secrecy throughput and secrecy energy efficiency by jointly determining the values of $R_s$, $R_e$, and $M$. 
To this end, we first analyze the COP and SOP for each transmission scheme, which will lay a solid foundation for the subsequent optimization.

\section{COP and SOP Analysis}
In this section, we derive the COP $\mathcal{O}_{co}$ and the SOP $\mathcal{O}_{so}$ for the three cooperative transmission schemes described in Sec. II-B, namely, DBF, FOT, and BSR schemes, respectively. 
For ease of notation, we will drop the superscripts $\{ \rm D,F,B\}$ when discussing the specific scheme. 

\subsection{DBF Scheme}

\subsubsection{COP}
In the DBF scheme, all the $K$ SBSs transmit the same file simultaneously, and a connection outage event occurs if $\log_2\left(1+\gamma_o\right)<R_t$, where $\gamma_o$ is the SNR of the typical user given in \eqref{snr_o_db} and $R_t$ is the wiretap codeword rate. Let $\beta_t=2^{R_t}-1$ denote the  threshold SNR for connection outage.
Then, the COP for the DBF scheme can be expressed as 
\begin{align}\label{pco_dbf}
\mathcal{O}_{co}&= \mathbb{P}\left\{\gamma_o< \beta_t\right\}=\mathbb{P}\left\{\left|\sum_{k=1}^K |h_{s_k,o}|r_{s_k,o}^{-\alpha/2}\right|^2< \frac{\beta_t}{P_s}\right\}.
\end{align}
Since the term $\left|\sum_{k=1}^K |h_{s_k,o}|r_{s_k,o}^{-\alpha/2}\right|^2$ in \eqref{pco_dbf} is the squared sum of independent and non-identically distributed Rayleigh random variables, it is intractable to derive a closed-form expression for the exact $\mathcal{O}_{co}$.
Instead, we provide an integral representation for $\mathcal{O}_{co}$ in the following theorem.
\begin{theorem}\label{theorem_pco_dbf}
The COP for the DBF scheme is given by
\begin{align}\label{pco_dbf_exact}
\mathcal{O}_{co} = \left(\frac{2\beta_t}{P_s}\right)^K
&\int\limits_{ \substack{
		y_1\ge 0, \cdots, y_K\ge 0\\
		\sum_{k=1}^K y_k<1}}e^{-\frac{\beta_t}{P_s}\sum_{k=1}^K{r_{s_k,o}^{\alpha}y_k^2}}\times\nonumber\\
&\prod_{k=1}^{K}\left({r_{s_k,o}^{\alpha}y_k}\right)dy_1,\cdots, dy_K.
\end{align} 
\end{theorem}
\begin{proof}
Please refer to Appendix \ref{appendix_theorem_pc_db}.
\end{proof}

To facilitate the analysis, we give a closed-form expression for $\mathcal{O}_{co}$ for the high SNR regime.
\begin{corollary}\label{corollary_pco_dbf}
	At the high SNR regime where $P_s\rightarrow\infty$, the COP $\mathcal{O}_{co}$ in \eqref{pco_dbf_exact} approaches
	\begin{align}\label{pco_dbf_asymptotic}
{\mathcal{O}}_{co}= \frac{2^K}{\Gamma(2K+1)}\left(\frac{\beta_t}{P_s}\right)^K\prod_{k=1}^{K}{r_{s_k,o}^{\alpha}}.
	\end{align}
\end{corollary}
\begin{proof}
	The result follows easily by invoking $\lim_{P_s\rightarrow\infty} e^{-\left({\beta_t}/{P_s}\right)\sum_{k=1}^K{r_{s_k,o}^{\alpha}y_k^2}}\rightarrow 1$ and \cite[Eqn. (4.634)]{Gradshteyn2007Table} with \eqref{pco_dbf_exact}.
\end{proof}

Corollary \ref{corollary_pco_dbf} shows that the COP decreases with the SBS power $P_s$ exponentially and increases with the codeword rate $R_t$, the SBS-user distance $r_{s_k,o}$, and the path-loss exponent $\alpha$. The asymptotic $\mathcal{O}_{co}$ in \eqref{pco_dbf_asymptotic} further reveals that, due to the joint transmission, the DBF scheme can achieve full diversity, i.e., a $K$-order diversity gain, where the diversity gain is defined as the rate of decay to zero of the COP for the high SNR regime, i.e., $- \lim_{P_s\rightarrow\infty}{\ln\mathcal{O}_{co}}/{\ln P_s}$ \cite{Zheng2003Diversity}.

\subsubsection{SOP}
A secrecy outage event happens if $\log_2(1+\max_{e_j\in\Phi_e}{\gamma_{e_j}})>R_e$, where $\gamma_{e_j}$ is the SNR of the $j$-th eavesdropper given in \eqref{snr_e_db} and $R_e$ is the rate redundancy. 
Let $\beta_e=2^{R_e}-1$ denote the threshold SNR for secrecy outage. Then, the SOP can be calculated as follows:
\begin{align}\label{pso_db}
&\mathcal{O}_{so}= \mathbb{P}\left\{\max_{e_j\in\Phi_e}\gamma_{e_j}> \beta_e\right\}\nonumber\\
& =1-\mathbb{E}_{\Phi_e}\left[\prod_{e_j\in\Phi_e}\mathbb{P}\left\{P_s\left|\sum_{k=1}^K \frac{h^{\dagger}_{s_k,o}h_{s_k,e_j}}{|h_{s_k,o}|}r_{s_k,e_j}^{-\alpha/2}\right|^2\le\beta_e\right\}\right]\nonumber\\
&\stackrel{\mathrm{(a)}}= 1-\mathbb{E}_{\Phi_e}\left[\prod_{e_j\in\Phi_e}\left(1-e^{-\frac{\beta_e/P_s}{\sum_{k=1}^Kr_{s_k,e_j}^{-\alpha} }}\right)\right]\nonumber\\
&\stackrel{\mathrm{(b)}}
=1-\exp\left(-\lambda_e\int_0^\infty\int_0^{2\pi}e^{-\frac{\beta_e/P_s}{\sum_{k=1}^Kr_{s_k,e}^{-\alpha} }}rdrd\theta\right),
\end{align}
where step $\rm (a)$ holds as the term $\left|\sum_{k=1}^K \frac{h^{\dagger}_{s_k,o}h_{s_k,e_j}}{|h_{s_k,o}|}r_{s_k,e_j}^{-\alpha/2}\right|^2$ obeys the exponential distribution with mean $\sum_{k=1}^Kr_{s_k,e_j}^{-\alpha}$, and step $\rm (b)$ follows from the probability generating functional (PGFL) over a PPP {which implements $\mathbb{E}_{\Phi}\left[\prod_{x\in\Phi}f(x)\right]=\exp(-\lambda\int_{\mathbb{R}^2}[1-f(x)]dx)$ for a PPP $\Phi$ of density $\lambda$ and a real valued function $f(x): \mathbb{R}^2\rightarrow [0,1]$} \cite[Sec. 4.3.6]{Chiu2013Stochastic}. Note that, we have $r_{s_k,e}=\sqrt{r_{s_k}^2+r^2-2r_{s_k}r\cos(\theta_{s_k}-\theta)}$ in step $\rm (b)$. Although $\mathcal{O}_{so}$ in \eqref{pso_db} is not closed-form, the integral is fairly computational-convenient.
We can easily confirm that $\mathcal{O}_{so}$ increases with the eavesdropper density $\lambda_e$ and the SBS  power $P_s$, and decreases with the rate redundancy $R_e$.

\subsection{FOT Scheme}
\subsubsection{COP}
In the FOT scheme, the requested file is divided into $K$ disjoint partitions, and connection outage takes place if not all the partitions are decoded correctly by the typical user. In other words, the COP can be interpreted as the probability that at least one $k\in\{1,\cdots,K\}$ satisfies $\log_2\left(1+\gamma_{s_k,o}\right)<R_t$, where $\gamma_{s_k,o}$ is the SNR of the $k$-th main channel given in \eqref{snr_o_ot}. Hence, a closed-form expression for the COP can be provided as below:
\begin{align}\label{pco_fot}
&\mathcal{O}_{co}
=1 - \mathbb{P}\left\{\bigcap^K_{k=1}\gamma_{s_k,o}\ge \beta_{t}\right\}\nonumber\\
&=1-\prod_{k=1}^K\mathbb{P}\left\{\frac{\left| h_{s_k,o}\right|^2}{r_{s_k,o}^{\alpha}}\ge \frac{\beta_t}{KP_s}\right\}=1- e^{-\frac{\beta_{t}}{KP_s }\sum_{k=1}^Kr_{s_k,o}^{\alpha}}.
\end{align}

At the high SNR regime with $P_s\rightarrow\infty$, $\mathcal{O}_{co}$ approaches $\frac{\beta_{t}}{KP_s }\sum_{k=1}^Kr_{s_k,o}^{\alpha}$, which indicates that the FOT scheme can only achieve a $1$-order diversity gain. The following proposition states that, due to the diversity loss, the FOT scheme is inferior to the DBF scheme in terms of reliability.
{ \begin{proposition}
The FOT scheme produces a larger COP than that of the DBF scheme. 
\end{proposition}
\begin{proof}
The second equality in \eqref{pco_fot} yields $\mathcal{O}_{co}
	>1-\mathbb{P}\left\{\min\limits_{k=1,\cdots,K}{K\left| h_{s_k,o}\right|^2}{r_{s_k,o}^{-\alpha}}\ge {\beta_t}/{P_s}\right\}>1-\mathbb{P}\left\{\left|\sum_{k=1}^K {|h_{s_k,o}|}{r_{s_k,o}^{-\alpha/2}}\right|^2\ge {\beta_t}/{P_s}\right\}$, with the last term equal to the COP given in \eqref{pco_dbf}. 
\end{proof}} 

\subsubsection{SOP}
Perfect secrecy is violated if an arbitrary partition of the requested file is intercepted by the eavesdroppers.
Then, the SOP can be defined as the complement of the probability that $\log_2\left(1+\gamma_{s_k,e_j}\right)\le R_e$ holds for all $k\in\{1,\cdots,K\}$, where $\gamma_{s_k,e_j}$ is the SNR of the wiretap channel from the $k$-th SBS to the $j$-th eavesdropper given in \eqref{snr_e_ot}. 
Hence, we have
\begin{align}\label{pso_ot}
&\mathcal{O}_{so}
=1-\mathbb{E}_{\Phi_e}\left[\prod_{e_j\in\Phi_e}\mathbb{P}\left\{\bigcap^K_{k=1}\gamma_{s_k,e_j}\le \beta_{e}\right\}\right]\nonumber\\
&=1-\mathbb{E}_{\Phi_e}\left[\prod_{e_j\in\Phi_e}\prod^K_{k=1}\mathbb{P}\left\{\frac{\left|h_{s_k,e_j}\right|^2}{r_{s_k,e_j}^{\alpha}}\le \frac{\beta_{e}}{KP_s} \right\}\right]\nonumber\\
&=1-\mathbb{E}_{\Phi_e}\left[\prod_{e_j\in\Phi_e}\prod^K_{k=1}\left(1-e^{-\frac{\beta_{e}r_{s_k,e_j}^{\alpha}}{KP_s} }\right)\right]\nonumber\\
&=1-\exp\left(\!-\lambda_e\!\int_0^\infty\!\!\int_0^{2\pi}\!\left[1-\!\prod^K_{k=1}\left(1-e^{-\frac{\beta_{e}r_{s_k,e}^{\alpha}}{KP_s}}\right)\right]rdrd\theta\right).
\end{align}
{ \begin{proposition}\label{proposition_sop_fot}
The FOT scheme provides a larger 	SOP than that of the DBF scheme.
\end{proposition}
\begin{proof}
Let $\omega_d = e^{-{\beta_e}/\left({P_s}{\sum_{k=1}^Kr_{s_k,e}^{-\alpha} }\right)}$ in \eqref{pso_db} and $\omega_o = 1-\prod^K_{k=1}\left(1-e^{-{\beta_e}/\left({P_s}Kr_{s_k,e}^{-\alpha}\right) }\right)$ in \eqref{pso_ot}. Now that the SOPs for the DBF and FOT schemes share similar forms, the proof can be completed by proving that $\omega_o>1-\min\limits_{k=1,\cdots,K}\left(1-e^{-{\beta_e}/\left({P_s}Kr_{s_k,e}^{-\alpha}\right)  }\right)=e^{-{\beta_e}/\left({KP_s}\max\limits_{k=1,\cdots,K}r_{s_k,e}^{-\alpha}\right) }>\omega_d$.
\end{proof}}

Proposition \ref{proposition_sop_fot} illustrates that although the DBF scheme superposes $K$ identical signals at the eavesdropper, it still can provide a higher level of secrecy than that of the FOT scheme. This is mainly because that the FOT scheme triggers $K$ times information leakage to the eavesdropper.

\subsection{BSR Scheme}\label{bsr_scheme}
\subsubsection{COP}
In the BSR scheme, the SBS having the highest main channel capacity is selected to serve the typical user, and connection outage occurs if $\log_2(1+{\gamma}_{s_{k^*},o})<R_t$, where ${\gamma}_{s_{k^*},o}$ is the maximal SNR of the $K$ main channels given in \eqref{snr_o_ss} with index $k^*=\arg_{1\le k\le K}\max |h_{s_k,o}|^2r_{s_k,o}^{-\alpha}$. 
A closed-form expression for the COP can be given below:
\begin{align}\label{pco_bsr}
\mathcal{O}_{co} &=\mathbb{P}\left\{\max_{1\le k\le K} \frac{|h_{s_k,o}|^2}{r_{s_k,o}^{\alpha}} <\frac{\beta_t}{P_s}\right\}
=\prod_{k=1}^{K}\mathbb{P}\left\{ \frac{|h_{s_k,o}|^2}{r_{s_k,o}^{\alpha}}<\frac{\beta_t}{P_s}\right\} \nonumber\\&= \prod_{k=1}^{K}\left(1-e^{-\frac{\beta_{t}r_{s_k,o}^{\alpha}}{P_s}}\right).
\end{align}
	\begin{proposition}\label{proposition_pco_bsr1}
	The BSR scheme can achieve a $K$-order diversity gain. 
\end{proposition}
\begin{proof}
	Substituting the approximation $\lim_{P_s\rightarrow\infty}1-e^{-{\beta_{t}r_{s_k,o}^{\alpha}}/{P_s}}\approx {\beta_{t}r_{s_k,o}^{\alpha}}/{P_s}$ into \eqref{pco_bsr} yields $ \mathcal{O}_{co}=\left({\beta_t}/{P_s}\right)^K\prod_{k=1}^{K}{r_{s_k,o}^{\alpha}}$, which reveals a $K$-order diversity. 
\end{proof}
\begin{proposition}\label{proposition_pco_bsr2}
	 The BSR scheme gives a COP larger than that of the DBF scheme but less than that of the FOT scheme. 
\end{proposition}
\begin{proof}
Clearly, $\left|\sum_{k=1}^K |h_{s_k,o}|r_{s_k,o}^{-\alpha/2}\right|^2$ in \eqref{pco_dbf} is larger than  ${|h_{s_k,o}|^2}{r_{s_k,o}^{-\alpha}}$ in \eqref{pco_bsr}, which infers that the BSR scheme provides a larger COP than the DBF scheme does. Likewise, $\mathcal{O}_{co}$ in \eqref{pco_bsr} is less than $\min\limits_{k=1,\cdots,K}\left(1-e^{-{\beta_{t}r_{s_k,o}^{\alpha}}/{P_s}}\right)=1 - e^{-\beta_{t}\min\limits_{k=1,\cdots,K}{r_{s_k,o}^{\alpha}}/{P_s}}$, which further lies below \eqref{pco_fot}. This indicates that the BSR scheme yields a smaller COP than that of the FOT scheme.
\end{proof}

Proposition \ref{proposition_pco_bsr2} declares that although both the DBF and BSR schemes can achieve full diversity, the former can still provide a higher level of reliability since it attains an additional performance gain from the coherent superposition of the received signals.

\subsubsection{SOP}
Content delivery is divided into two hops, and then the content is considered secure only if both the two hops are secured.  Since an eavesdropper decodes the two-hop signals individually, the overall SOP can be characterized as a combination of the individual SOPs for the two hops. Recall the received SNRs at the $j$-th eavesdropper for the two hops, i.e., $\gamma_{b,e_j}$ in \eqref{snr_e_ss1} and $\gamma_{s_{k^*},e_j}$ in \eqref{snr_e_ss2}, and the SOP can be displayed as 
\begin{equation}\label{pso_ss}
\mathcal{O}_{so} =1-\mathbb{E}_{\Phi_e}\left[\prod_{e_j\in\Phi_e}\left(1-\mathcal{O}^{(1)}_{so,e_j}\right)\left(1-\mathcal{O}^{(2)}_{so,e_j}\right)\right] ,
\end{equation}
where $\mathcal{O}^{(1)}_{so,e_j}\triangleq\mathbb{P}\left\{\gamma_{b,e_j}>\beta_e\right\}=e^{-\beta_er_{b,e_j}^{\alpha}/P_m}$ and $\mathcal{O}^{(2)}_{so,e_j}\triangleq\mathbb{P}\left\{\gamma_{s_{k^*},e_j}>\beta_e\right\}=e^{-\beta_er_{s_{k^*},e_j}^{\alpha}/P_s}$ are the SOPs for the two hops.
Invoking the PGFL over a PPP with \eqref{pso_ss} yields
\begin{align}\label{pso_ss_int}
\mathcal{O}_{so} =1-\exp\bigg(-&\lambda_e\int_0^\infty\int_0^{2\pi}\left[1-\left(1-e^{-\beta_er_{b,e}^{\alpha}/P_m}\right)\times\right.\nonumber\\
&\left.\left(1-e^{-\beta_er_{s_k^*,e}^{\alpha}/P_s}\right)\right]rdrd\theta\bigg).
\end{align}

In the following proposition, we uncover that whether the BSR scheme can outperform the DBF and FOT schemes or not depends heavily on the secrecy level of the first hop. 
\begin{proposition}\label{proposition_bsr}
	The BSR scheme can provide a higher secrecy level compared with the DBF and FOT schemes if a sufficiently low $\mathcal{O}^{(1)}_{so,e_j}$ is ensured in the first hop; in contrast, it will achieve the poorest secrecy performance if $\mathcal{O}^{(1)}_{so,e_j}$ is exceedingly large. 
\end{proposition}
\begin{proof}
	For the case  $\mathcal{O}^{(1)}_{so,e_j}=e^{-\beta_er_{b,e_j}^{\alpha}/P_m}\rightarrow 0$, the SOPs in \eqref{pso_db}, $\eqref{pso_ot}$, and $\eqref{pso_ss_int}$ share similar forms. We can prove that $e^{-\beta_er_{s_k^*,e}^{\alpha}/P_s}$ in $\eqref{pso_ss_int}$ is less than both $e^{-{\beta_e}/\left({P_s\sum_{k=1}^Kr_{s_k,e}^{-\alpha} }\right)}$ in \eqref{pso_db} and $1-\prod^K_{k=1}\left(1-e^{-\beta_{e}r_{s_k,e}^{\alpha}/(KP_s)}\right)$ in \eqref{pso_ot}, which indicates that the BSR scheme yields a smallest SOP among the three schemes. For the case $\mathcal{O}^{(1)}_{so,e_j}\rightarrow 1$, the SOP for the BSR scheme approaches one, which is larger than those of the other two schemes.
\end{proof}

The reason behind Proposition \ref{proposition_bsr} is that the backhaul process is one major bottleneck of the secrecy performance. As will be detailed in the next two sections, the MBS power $P_m$ and the backhaul probability $p_t^{\rm B}$ play a vital role in improving throughput and energy efficiency. 

Note that due to the random mobility of the eavesdroppers, their locations in the two hops can be approximately regarded as two independent PPPs $\Phi^{(1)}_{e}$ and $\Phi^{(2)}_{e}$ with the same density $\lambda_e$. 
In this way, the SOP can be equivalently presented in a more concise form given blow:
\begin{align}\label{pso_bsr_app}
&\mathcal{O}_{so} =1-\prod_{i=1,2}\mathbb{E}_{\Phi^{(i)}_e}\left[\prod_{e_j\in\Phi^{(i)}_e}\left(1-\mathcal{O}^{(i)}_{so,e_j}\right)\right] \nonumber\\
&= 1 - \exp\left(-{\pi\lambda_e}\Gamma\left(1+\frac{2}{\alpha}\right)\left({P_m^{\frac{2}{\alpha}}}+{P_s^{\frac{2}{\alpha}}}\right)\beta_e^{-\frac{2}{\alpha}}\right),
\end{align}
where the last equality follows from the PGFL over a PPP along with \cite[Eqn. (3.216.1)]{Gradshteyn2007Table}.

In Fig. \ref{PCO} and Fig. \ref{PSO}, we depict the COP $\mathcal{O}_{co}$ and the SOP $\mathcal{O}_{so}$ versus the SBS power $P_s$ for the three transmission schemes.
The results of Monte-Carlo simulations match well with the
theoretical values. The asymptotic $\mathcal{O}_{co}$ in \eqref{pco_dbf_asymptotic} and the approximated $\mathcal{O}_{so}$ in \eqref{pso_bsr_app} have high accuracies with their respective exact values. As expected, as $P_s$ increases, $\mathcal{O}_{co}$ decreases and $\mathcal{O}_{so}$ increases. Just as proved in the last three subsections, the DBF scheme outperforms the other two while the FOT scheme provides the poorest reliability performance. For transmission secrecy, the DBF scheme always surpasses the FOT scheme; whereas the BSR scheme yields the largest SOP at the low $P_s$ regime but achieves the smallest SOP at the high $P_s$ regime. The underlying reason is that, compared with the DBF and FOT schemes where $K$ SBSs transmit signals to the eavesdroppers, the BSR scheme only selects a single SBS in the second hop, which makes the wiretapping more difficult if only the first-hop secrecy can be promised. 
\begin{figure}[!t]
	\centering
	\includegraphics[width = 3.0in]{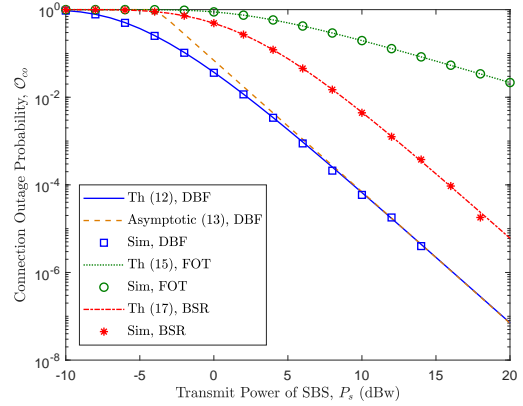}
	\caption{COP $\mathcal{O}_{co}$ vs. $P_s$, with $K=3$ and $\beta_t = 1$. Throughout the experiments in this paper, for simplicity, we place the typical user, the nearest SBS, and the MBS along a vertical line, and deploy all the SBSs along a horizontal line with an identical distance $r_s$. Unless otherwise specified, we always set $r_{b,s_1}=2$, $r_{s_1,o}=1$, $r_s=0.5$, and $\alpha=4$.}
	\label{PCO}
\end{figure}
\begin{figure}[!t]
	\centering
	\includegraphics[width = 3.0in]{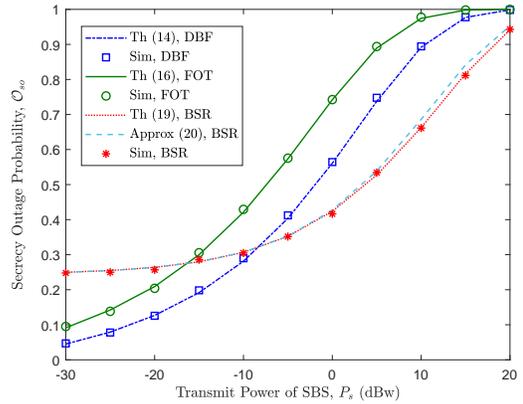}
	\caption{SOP  $\mathcal{O}_{so}$ vs. $P_s$, with $K=5$, $P_m=0$ dBw, $\lambda_e = 0.1$, and $\beta_e = 1$. }
	\label{PSO}
\end{figure}

\section{Secrecy Throughput Maximization}\label{sec_st_max}
{ In this section, we will jointly design the optimal secrecy rate $R_s$ and rate redundancy $R_e$ of the wiretap code and the allocation $M$ of the hybrid caching policy to maximize the overall secrecy throughput $\bar\Psi$ given in \eqref{def_st}. This optimization problem can be formulated as follows:
\begin{subequations}
	\begin{align}\label{overall_st_max_s}
	\max_{R_s^i,R_e^i,M,\forall i\in\{ \rm D,F,B\}} ~&\bar\Psi=\sum_{i\in\{ \rm D,F,B\}}p_t^{i}\left(1-\mathcal{O}^i_{co}\right)R_s^i,\\
	~~~{\rm s.t.} \label{overall_st_max_c1}
	~&R_s^i\geq  0, R_e^i\ge 0,\\
	\label{overall_st_max_c2}
	~&\mathcal{O}_{so}^i\le\epsilon,\\
	\label{overall_st_max_c3}
	~&0\le M \le L,
	\end{align}
\end{subequations}
where \eqref{overall_st_max_c1}, \eqref{overall_st_max_c2}, and \eqref{overall_st_max_c3} are the constraints of the wiretap code rates, the maximal tolerable SOP, and the caching capacity, respectively, and $p_t^{i}$, $\mathcal{O}_{co}^i$, and $\mathcal{O}_{so}^i$ denote the probability of adopting the transmission scheme $i\in\{\rm D,F,B\}$ and the corresponding COP and SOP, respectively. }

Observing that only $p_t^i$ in $\bar\Psi$ depends on $M$, the primary problem \eqref{overall_st_max_s} can be decomposed into two steps: 1) first maximizing $\Psi^i=\left(1-\mathcal{O}_{co}^i\right)R_s^i$ subject to $\mathcal{O}^i_{so}\le\epsilon$ over $R_s^i$ and $R_e^i$ for scheme $i\in\{ \rm D,F,B\}$; 2) then maximizing $\bar\Psi$ over $M$. In what follows, we execute the
optimization procedure step by step.
To begin with, we formulate the first-step problem uniformly as below:
\begin{equation}\label{st_max_s}
\max_{R_s\ge0,R_e\ge0} ~\Psi=\left(1-\mathcal{O}_{co}\right)R_s,
~~{\rm s.t.} ~
 \mathcal{O}_{so}\le\epsilon.
\end{equation}

Intuitively, neither a too small nor a too large $R_s$ can lead to a high $\Psi$ due to the reverse behavior of $\mathcal{O}_{co}$.
Hence, $R_s$ should be carefully chosen to balance the rate and outage.
Once $R_s$ is fixed, $\Psi$ monotonically decreases with $R_e$. This suggests that we should set $R_e$ as small as possible in order to maximize $\Psi$. Denote the minimal $R_e$ as $R_e^{\circ}=\log_2(1+\beta_e^{\circ})$. Obviously, we should satisfy  $\mathcal{O}_{so}(\beta_e^{\circ})=\epsilon$ since $\mathcal{O}_{so}(\beta_e)$ decreases with $\beta_e$. The value of $\beta_e^{\circ}$ can be quickly obtained via a bisection search, which is $\beta_e^{\circ}\triangleq\mathcal{O}_{so}^{-1}(\epsilon)$, where $\mathcal{O}_{so}^{-1}(\epsilon)$ is the inverse function of $\mathcal{O}_{so}(\beta_e)$. With $R_t = R_e^{\circ}+R_s\Rightarrow\beta_t=\beta_e^{\circ}+(1+\beta_e^{\circ})\beta_s$, problem \eqref{st_max_s} can be recast into 
\begin{align}\label{st_max}
\max_{\beta_s\ge 0}~ \Psi=(1-\mathcal{O}_{co})\log_2(1+\beta_s).
\end{align}

\subsection{Optimal $\beta_s$ for the DBF Scheme}\label{sec_st_max_db}
Substituting the COP $\mathcal{O}_{co}$ in \eqref{pco_dbf_asymptotic} into \eqref{st_max} yields
\begin{align}\label{st_max_db}
\max_{\beta_s\ge 0}~ \Psi=\left[1-A_1(\beta_s+B_1)^K\right]\log_2(1+\beta_s),
\end{align}
where $A_1=\frac{2^K}{\Gamma(2K+1)}\left(\frac{1+\beta_e^{\circ}}{P_s}\right)^K\prod_{k=1}^{K}{r_{s_k,o}^{\alpha}}$ and $B_1=\frac{\beta_e^{\circ}}{1+\beta_e^{\circ}}$.
The solution to problem \eqref{st_max_db} is given in the following theorem.
\begin{theorem}\label{theorem_st_db}
The secrecy throughput $\Psi$ for the DBF scheme given in \eqref{st_max_db} is a concave function of $\beta_s$, and the optimal $\beta_s^*$ that maximizes $\Psi$ is characterized by 
\begin{equation}\label{opt_bs_db}
\frac{d\Psi}{d\beta_s^*} = 0,
\end{equation}
i.e., it is the unique zero-crossing of the derivative $\frac{d\Psi}{d\beta_s}$ with
\begin{equation}\label{opt_bs_db_eqn}
\frac{d\Psi}{d\beta_s} = \frac{1-A_1(\beta_s+B_1)^K}{(1+\beta_s)\ln2}-A_1K(\beta_s+B_1)^{K-1}\log_2(1+\beta_s).
\end{equation}
\end{theorem}
\begin{proof}
We can easily show that $\frac{d\Psi}{d\beta_s}$ in \eqref{opt_bs_db_eqn} monotonically decreases with $\beta_s$.
Hence, $\Psi$ in \eqref{st_max_db} is concave on $\beta_s$. 
Next, we prove that $\frac{d\Psi}{d\beta_s}$ is positive at $\beta_s=0$ and is negative as $\beta_s\rightarrow\infty$. 
Therefore, there exists a unique zero-crossing of  $\frac{d\Psi}{d\beta_s}$, which is the solution to problem \eqref{st_max_db}.
\end{proof}

Due to the concavity of $\Psi$ on $\beta_s$, the optimal $\beta_s^*$ can be efficiently calculated using the Newton's method with \eqref{opt_bs_db}.
Although $\beta_s^*$ appears in an implicit form, some insights into its behavior can still be developed in the following corollary.
\begin{corollary}\label{corollary_opt_bs_db}
The optimal $\beta_s^*$ satisfying \eqref{opt_bs_db} increases with the threshold SOP $\epsilon$ and decreases with the eavesdropper density $\lambda_e$, the SBS-user distance $r_{s_k,o}$, and the path-loss exponent $\alpha$.
\end{corollary}
\begin{proof}
Let $Z(\beta_s^*,A_1)$ denote $\frac{d\Psi}{d\beta_s} $ in \eqref{opt_bs_db_eqn} with $\beta_s$ replaced by $\beta_s^*$. Then, the derivative $\frac{d\beta_s^*}{dA_1}$ can be calculated by using the derivative rule for implicit functions \cite{Zheng2015Multi} with \eqref{opt_bs_db},
\begin{equation}
\frac{d\beta_s^*}{dA_1}= -\frac{\partial Z/\partial A_1}{\partial Z/\partial \beta_s^*}<0.
\end{equation}
By realizing that $A_1$ is an increasing function of $r_{s_k,o}$, $\alpha$, and $\beta_e^{\circ}$, and meanwhile $\beta_e^{\circ}$ decreases with $\epsilon$ and increases with $\lambda_e$, the proof can be completed.
\end{proof}

\subsection{Optimal $\beta_s$ for the FOT Scheme}\label{sec_st_max_ot}
Plugging the COP $\mathcal{O}_{co}$ in \eqref{pco_fot} into \eqref{st_max}, we have 
\begin{align}\label{st_max_ot}
\max_{\beta_s\ge 0}~ \Psi=A_2e^{-B_2\beta_s}\log_2(1+\beta_s),
\end{align}
where $A_2 = e^{-\left({\beta_e^{\circ}}/{P_s}\right)\sum_{k=1}^Kr_{s_k,o}^\alpha}$ and $B_2=\frac{1+\beta_e^{\circ}}{P_s}\sum_{k=1}^Kr_{s_k,o}^\alpha$. 
The solution to problem \eqref{st_max_ot} is presented by the following theorem.
\begin{theorem}\label{theorem_opt_bs_fot}
The secrecy throughput $\Psi$ for the FOT scheme in \eqref{st_max_ot} is quasi-concave on $\beta_s$, and the optimal $\beta_s^*$ maximizing $\Psi$ is 
the unique zero-crossing of the following derivative $\frac{d\Psi}{d\beta_s}$,
\begin{equation}\label{opt_bs_ot_eqn}
\frac{d\Psi}{d\beta_s}= -\frac{A_2e^{-B_2\beta_s}}{\ln2}\left(B_2\ln(1+\beta_s)-\frac{1}{1+\beta_s}\right).
\end{equation}
\end{theorem}
\begin{proof}
It can be proved that the two boundaries are $\frac{d\Psi}{d\beta_s}|_{\beta_s=0}=A_2>0$ and $\frac{d\Psi}{d\beta_s}|_{\beta_s\rightarrow\infty}<0$.
Due to the continuity of $\Psi$, there is at least one zero-crossing of  $\frac{d\Psi}{d\beta_s}$.
Denote an arbitrary one as $\beta_s^*$ such that $B_2\ln(1+\beta_s^*)-\frac{1}{1+\beta_s^*}=0$.
Then, we have the second derivative $\frac{d^2\Psi}{d\beta_s^2}$ at $\beta_s=\beta_s^*$,
\begin{equation}
\frac{d^2\Psi}{d\beta_s^2}|_{\beta_s=\beta_s^*}=-\frac{A_2e^{-B_2\beta_s}}{(1+\beta_s)\ln2}\left(B_2+\frac{1}{1+\beta_s}\right)<0.
\end{equation}
This indicates that $\Psi$ is a quasi-concave function of $\beta_s$ \cite{Boyd2004Convex}.
Hence, $\beta_s^*$ is the unique zero-crossing of $\frac{d\Psi}{d\beta_s}$, i.e., it is the solution to problem \eqref{st_max_ot}.
\end{proof}

Note that the derivative $\frac{d\Psi}{d\beta_s}$ in \eqref{opt_bs_ot_eqn} is initially positive and then negative as $\beta_s$ increases, and hence $\beta_s^*$ can be easily derived through a bisection search with the equation $\frac{d\Psi}{d\beta_s} = 0$.
Furthermore, following similar steps as Corollary \ref{corollary_opt_bs_db}, the same conclusions can be made on the relationship between the optimal $\beta_s^*$ and the parameters $\epsilon$, $r_{s_k,o}$, and $\alpha$.

\subsection{Optimal $\beta_s$ for the BSR Scheme}\label{sec_st_max_ss}
Inserting the COP $\mathcal{O}_{co}$ in \eqref{pco_bsr} into \eqref{st_max} arrives at
\begin{align}\label{st_max_ss}
\max_{\beta_s\ge 0} \Psi=\frac{\log_2(1+\beta_s)}{2}\left[1-\prod_{k=1}^{K}\left(1-A_{3,k}e^{-B_{3,k}\beta_s}\right)\right],
\end{align}
where $A_{3,k} = e^{-{\beta_e^{\circ}}r_{s_k,o}^\alpha/{P_s}}$ and $B_{3,k}=\frac{1+\beta_e^{\circ}}{P_s}r_{s_k,o}^\alpha$. Note that the factor $\frac{1}{2}$ exists due to a two-hop delivery process.
Although it is difficult to determine the concavity of $\Psi$ with respect to $\beta_s$, one can easily verify that the derivative $\frac{d\Psi}{d\beta_s}$ is first positive and then negative as $\beta_s$ increases, which implies  that there exists a unique zero-crossing of $\frac{d\Psi}{d\beta_s}$.
Hence, the optimal $\beta_s^*$ maximizing $\Psi$ can be numerically computed by setting  $\frac{d\Psi}{d\beta_s^*}$ to zero. In the following proposition, we further provide a sub-optimal $\beta_s^{\circ}$ that can maximize a lower bound for $\Psi$ in \eqref{st_max_ss}. 
\begin{proposition}
	The value of $\Psi$ in \eqref{st_max_ss} is lower bounded by $ \Psi^{\circ}=\frac{\log_2(1+\beta_s)}{2}A_{3,k^{\circ}}e^{-B_{3,k^{\circ}}\beta_s}$, with index $k^{\circ}=\arg_{k=1,\cdots,K}\min r_{s_k,o}$. The optimal $\beta_s^{\circ}$ that maximizes $\Psi^{\circ}$ shares the same form as $\beta_s^*$ shown in \eqref{opt_bs_ot_eqn}, simply with $A_2$ and $B_2$ therein replaced with $A_{3,k^{\circ}}$ and $B_{3,k^{\circ}}$, respectively.
\end{proposition}
\begin{proof}
	We obtain the lower bound $\Psi^{\circ}$ from \eqref{st_max_ss} by realizing that $\prod_{k=1}^{K}\left(1-A_{3,k}e^{-B_{3,k}\beta_s}\right)<A_{3,k^{\circ}}e^{-B_{3,k^{\circ}}\beta_s}$. Then, the proof can be completed following the proof of Theorem \ref{theorem_opt_bs_fot}.
\end{proof}

In Fig. \ref{ST_RS}, we plot the secrecy throughput $\Psi$ as a function of the secrecy rate $R_s$. As proved in the above three subsections, $\Psi$ initially increases and then decreases with $R_s$, and there is a unique $R_s$ maximizing $\Psi$. We show that by exploiting the instantaneous CSI of the main channels, the DBF scheme attains a significant throughput gain over the other two schemes.
Besides, the FOT scheme is superior to the BSR scheme at the low $R_s$ regime whereas becomes inferior at the high $R_s$ regime. The cause behind is that, due to the bottleneck of the first-hop secrecy, the BSR scheme requires a larger rate redundancy $R_e$ than the FOT scheme does. Hence, at the low $R_s$ regime, the codeword rate $R_t = R_s+R_e$ and the resulting COP for the BSR scheme become remarkably larger than those for the FOT scheme. However, at the high $R_s$ regime, the superiority of the BSR scheme in terms of reliability is demonstrated, which counterbalances the adverse impact of the first-hop secrecy. We also observe that, $\Psi$ improves for a larger acceptable SOP $\epsilon$, and the optimal $R_s$ increases accordingly, just as indicated in Corollary \ref{corollary_opt_bs_db}.

\begin{figure}[!t]
	\centering
	\includegraphics[width = 3.0in]{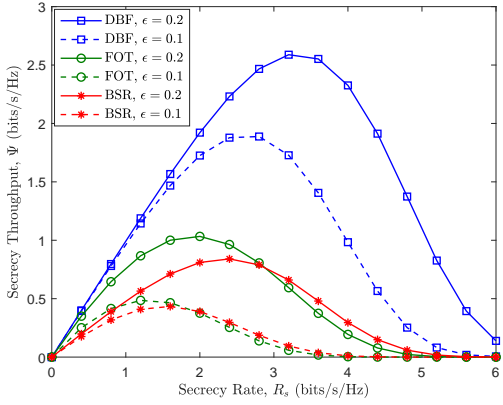}
	\caption{$\Psi$ vs. $R_s$ for different values of $\epsilon$, with $K=2$, $P_m=10$ dBw, $P_s=10$ dBw, and $\lambda_e=0.01$. }
	\label{ST_RS}
\end{figure}

Fig. \ref{MAX_ST_OPT_RS} depicts the maximal secrecy throughput $\Psi^*$ for the proposed transmission schemes. The high accuracy of the theoretical results to the simulations is confirmed. As expected, throughput performance deteriorates for a larger eavesdropper density $\lambda_e$. Owing to the adaption of transmission rates, increasing the SBS power $P_s$ can always improve secrecy throughput. Nevertheless, as $P_s$ grows large, the throughput gain becomes negligible due to the SOP constraint. We also find that the FOT scheme outperforms the BSR scheme at the low SNR regime but becomes inferior at the high SNR regime, which are quite similar to the observations in Fig. \ref{ST_RS}. 

\begin{figure}[!t]
	\centering
	\includegraphics[width = 3.0in]{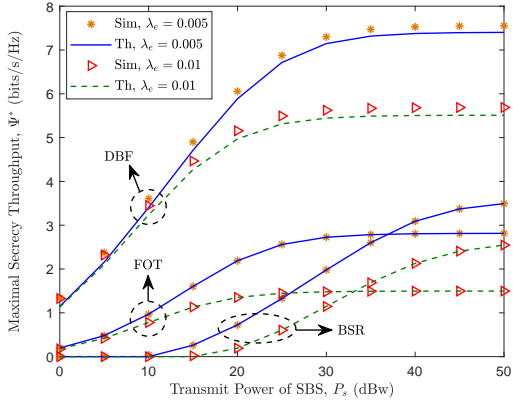}
	\caption{$\Psi^*$ vs. $P_s$ for different values of $\lambda_e$, with $K=3$, $P_m=40$ dBw, and $\epsilon=0.3$. }
	\label{MAX_ST_OPT_RS}
\end{figure}

\subsection{Optimal Caching Allocation Design}\label{sec_st_max_cache}
Having obtained the maximal secrecy throughput for each transmission scheme, we will determine the optimal caching allocation $M$ to maximize the overall secrecy throughput $\bar\Psi = \sum_{i\in\{ \rm D,F,B\}}p_t^{i} \Psi^{i}$.
{This subsection focuses on a limited caching capacity case where the total caching capacity is less than the file number, i.e., $KL<N$. In other words, whatever $M$ we choose, the cache miss event can happen if the popularity order of the requested file $F_n$ is $ n > M + K(L-M)$. The case $KL\ge N$ will be discussed in the next subsection.}

Before proceeding to the optimization problem, we first calculate the probabilities $p_t^{i}$ for $i\in\{\rm D,F,B\}$ as per the caching policy described in Sec. II, 
\begin{equation}\label{pt}
p_t^{\rm D} =\sum_{m=1}^M p_m,~ p_t^{\rm F} =\sum_{m=M+1}^{M+K(L-M)} p_m,~ p_t^{\rm B} =1-p_t^{\rm D}-p_t^{\rm F},\!
\end{equation}
with $p_m$ being the Zipf probability given in \eqref{zipf}.
For tractability, we approximate  $\sum_{m=1}^M p_m$ as
\begin{equation}\label{approx}
\sum_{m=1}^M p_m\approx \frac{1-(M+1)^{1-\tau}}{1-(N+1)^{1-\tau}},
\end{equation} which derives from $\sum_{m=1}^M m^{-\tau}\approx\int_{1}^{M+1}t^{-\tau}dt=\frac{1-(M+1)^{1-\tau}}{\tau - 1}$. 
Invoking \eqref{approx} with \eqref{pt}, the problem of maximizing $\bar \Psi$ over $M$ can be formulated as follow: 
\begin{equation}\label{max_st}
\max_{0\le M\le L}\bar\Psi = \frac{\Psi^{\rm DB} - \Psi^{\rm DF}(M+1)^{1-\tau}-\Psi^{\rm FB}(K_L-K_1M)^{1-\tau}}{1-(N+1)^{1-\tau}},
\end{equation}
where $K_L \triangleq KL+1$, $K_1\triangleq K-1$, and $\Psi^{\rm DF}\triangleq\Psi^{\rm D}-\Psi^{\rm F}$, $\Psi^{\rm FB}\triangleq\Psi^{\rm F}-\Psi^{\rm B}$, and $\Psi^{\rm DB}\triangleq \Psi^{\rm D}-\Psi^{\rm B}(N+1)^{1-\tau}$ reflect the throughput gaps between two transmission schemes. 
We always have $\Psi^{\rm DF}>0$ since the DBF scheme offers lower COP and SOP than those of the FOT scheme. 

The solution to problem \eqref{max_st} is provided below, with the proof relegated Appendix \ref{appendix_theorem_opt_m2}. 
\begin{theorem}\label{theorem_opt_m2}
With hybrid MPC and LCD caching placement, the optimal allocation of MPC caching $M_T^*$ that maximizes the overall secrecy throughput $\bar\Psi$ in \eqref{max_st} is given by
\begin{align}\label{opt_m_st}
M_T^{*} 
=\begin{cases}
L, & \Psi^{\rm DF}>K_1\Psi^{\rm FB},\\
0, & \Psi^{\rm DF}<{K_1}{K_L^{-\tau}}\Psi^{\rm FB},\\
\left\lceil L-\frac{L+1}{K\Lambda+1}\right\rceil,& \rm otherwise,\\
\end{cases}
\end{align}
where $\Lambda = \left[\left({K_1\Psi^{\rm FB}}/{\Psi^{\rm DF}}\right)^{\frac{1}{\tau}}-1\right]^{-1}\in\left({K_L^{-1}},\infty\right)$.
\end{theorem}

Theorem \ref{theorem_opt_m2} demonstrates that the optimal caching allocation relies significantly on the throughput difference between different transmission schemes. Specifically, i) if the throughput gain of the DBF scheme to the FOT scheme $\Psi^{\rm DF}$ is over $K_1$ times larger than the throughput gain of the FOT scheme to the BSR scheme $\Psi^{\rm FB}$, we have $M_T^*=L$, which suggests that MPC caching is beneficial and we should cache the same files of high popularities; 
ii) if the throughput gain $\Psi^{\rm FB}$ is remarkably larger, i.e., exceeding ${K_L^{\tau}}/{K_1}$ times than $\Psi^{\rm DF}$, we should switch to the LCD caching mode and store different partitions of files in different SBSs; 
iii) for a  moderate situation where ${K_1}{K_L^{-\tau}}\le{\Psi^{\rm DF}}/{\Psi^{\rm FB}}\le K_1$,  there exists an optimal allocation between MPC and LCD caching, which strikes a good balance between secrecy throughput and content diversity.
Additionally, we can prove that the optimal $M_T^*$ increases with the file popularity skewness $\tau$. This is because, as the content popularity becomes more concentrated (i.e., a larger $\tau$), the benefit from caching different files becomes limited in terms of throughput improvement.


{\subsection{Large Caching Capacity Cases}\label{st_large_case}
This subsection examines two large caching capacity cases where the total caching capacity $KL$ or even each SBS's storage capacity $L$ are not less than the number of files $N$, namely, $KL\ge N>L$ and $L\ge N$, respectively. 
Recall our hybrid caching policy, and we know that if we choose $M\le\frac{KL-N}{K-1}$ for the former case, all the $N$ files can be stored in the $K$ SBSs since  $M+K(L-M)\ge N$. In other words, the cache miss event can be avoided. For the latter case, a single SBS can store all the $N$ files, and therefore the cache miss event absolutely will not happen, or equivalently, the BSR scheme never will be activated. 
In what follows, we present the optimal caching allocation for the above two cases, respectively.
\begin{proposition}\label{opt_m_finity1}
	For the case $KL\ge N>L$, the optimal allocation $M_T^{\star}$ of MPC caching that maximizes  $\bar \Psi$ is given by $M_T^{\star} = \max\left\{\frac{KL-N}{K-1},M_T^*\right\}$, where $M_T^*$ is provided by Theorem \ref{theorem_opt_m2}.
\end{proposition}
\begin{proof}
First consider $M\le \frac{KL-N}{K-1}$, and we have $p_t^{\rm D} =\sum_{m=1}^M p_m$ and $p_t^{\rm F} = 1 - p_t^{\rm D} $. We can prove that $\bar \Psi = p_t^{\rm D}\Psi^{\rm D}+p_t^{\rm F}\Psi^{\rm F} $ increases with $M$ and reaches the maximum at $M = \frac{KL-N}{K-1}$. On the other hand, the optimal $M$ for $M>\frac{KL-N}{K-1}$ simply follows from Theorem \ref{theorem_opt_m2}. The proof can be completed by combing the results for the two situations.
\end{proof}
\begin{proposition}\label{opt_m_finity2}
For the case $L\ge N$, MPC caching can achieve a maximal $\bar \Psi$.
\end{proposition}
\begin{proof}
	The result follows easily by proving that $\bar \Psi$ increases with $M$ for $M<N$.
\end{proof}

Propositions \ref{opt_m_finity1} and \ref{opt_m_finity2} indicate that when the caching capacity and the number of files are comparable, the optimal caching allocation will be impacted by the file number, which differs from Theorem \ref{theorem_opt_m2}. More precisely, as the caching capacity increases, it is throughput-wise favorable to store more files via MPC caching. The fundamental cause lies in the growing probability of the cache hit event and the superiority of the DBF scheme itself in terms of throughput enhancement.

\begin{figure}[!t]
	\centering
	\includegraphics[width = 3.0in]{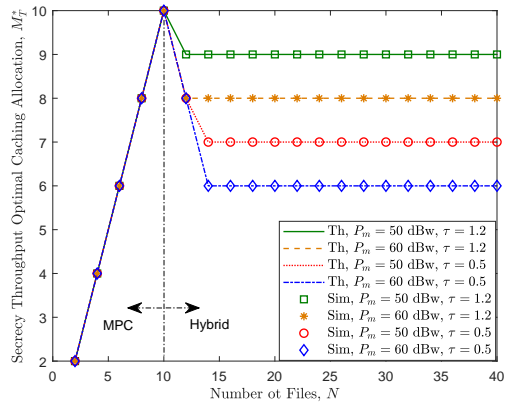}
	\caption{$M_T^*$ vs. $P_s$ for different values of $P_m$ and $\tau$, with $K=2$, $P_s = 20$ dBw, $\lambda_e=0.002$, $\epsilon=0.2$, and $L=10$.  }
	\label{OPT_MT}
\end{figure}
Fig. \ref{OPT_MT} describes the optimal allocation $M_T^*$ versus the number of files $N$.
The simulated optimal $M_T^*$ is obtained by exhaustive
search and matches well with 
the theoretical value, validating the accuracy of the approximation given in \eqref{approx}.
As proved by Propositions \ref{opt_m_finity1} and \ref{opt_m_finity2}, at the large caching capacity regime, i.e., $KL\ge N$, $M_T^*$ first increases linearly and then decreases linearly with $N$. As $N$ increases further, $M_T^*$ remains constant, which coincides with Theorem \ref{theorem_opt_m2}.
In addition, as explained for Theorem \ref{theorem_opt_m2}, $M_T^*$ decreases with the MBS power $P_m$ whereas increases with the content popularity skewness $\tau$.} 

\begin{figure}[!t]
	\centering
	\includegraphics[width = 3.0in]{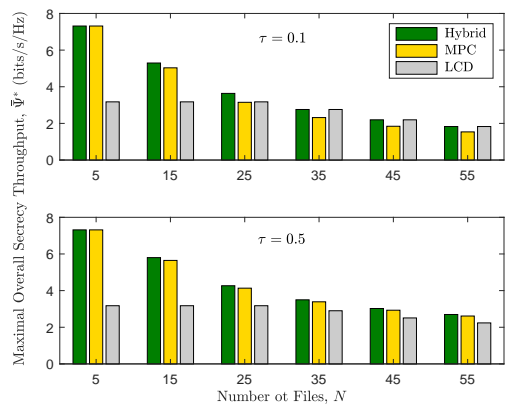}
	\caption{$\bar{\Psi}^*$ vs. $N$ for different values of $\tau$, with $K=3$, $P_m=60$ dBw, $P_s=25$ dBw, $\lambda_e = 0.002$, $\epsilon=0.2$, and $L = 10$. }
	\label{MAX_ST}
\end{figure}
Fig. \ref{MAX_ST} compares the maximal overall secrecy throughput $\bar{\Psi}^*$ for three different caching policies, namely, hybrid caching, MPC, and LCD caching, respectively. We show that hybrid caching always outperforms the other two, since it can strike a good balance between secrecy throughput and content diversity. We also find that $\bar{\Psi}^*$ decreases with $N$, since in this circumstance the cache miss probability increases, which weakens the benefit of caching. Additionally, as the file popularity becomes more concentrated (i.e., a larger $\tau$), MPC caching provides a larger $\bar{\Psi}^*$ than that of LCD caching and the performance is close to that of hybrid caching. This suggests that MPC can be an alternative solution for secrecy throughput maximization at the large $\tau$ regime.

\section{Secrecy Energy Efficiency Maximization}
This section aims to tackle the problem of maximizing the secrecy energy efficiency $\Omega$ defined in \eqref{def_see}, which can be formulated similarly as \eqref{overall_st_max_s} simply with the objective function $\bar\Psi$ replaced with $\Omega = {\bar\Psi}/{P_{\rm avg}}$, where $P_{\rm avg}$ is the total consumed power. Since $P_{\rm avg}$ only depends on the allocation of MPC caching $M$ rather than $R_s$ and $R_e$, this problem can be decomposed into two steps as described in Sec. IV where we have already obtained the maximal secrecy throughput $\Psi^{i}$ for transmission scheme $i\in\{ \rm D,F,B\}$ in the first step.
Hence, we only need to optimize $\Omega$ over $M$. Substitute $p_t^i$ in \eqref{pt} into \eqref{def_see}, and then this sub-problem can be formulated as follows:
\begin{equation}\label{max_see}
\max_{0\le M \le L}\Omega = \frac{\Psi^{\rm DB}\!-\!\Psi^{\rm DF}(M\!+\!1)^{1-\tau}\!-\!\Psi^{\rm FB}(K_L\!-\!K_1M)^{1-\tau}}{\Delta P_1+\Delta P_2(K_L\!-\!K_1M)^{1-\tau}},
\end{equation}
where $\Delta P_1\triangleq KP_s - (P_m+P_s)(N+1)^{1-\tau}$, $\Delta P_2\triangleq P_m - K_1P_s$, and $K_L$, $K_1$, $\Psi^{\rm DB}$, $\Psi^{\rm DF}$, and $\Psi^{\rm FB}$ are defined in \eqref{max_st}. 
{Note that for the large caching capacity cases $KL\ge N$ as discussed in Sec. IV-E, the optimal caching allocation follows easily from Propositions \ref{opt_m_finity1} and \ref{opt_m_finity2} with quite similar reasons behind. 
Hence, due to the page limitation, in this subsection we only examine the limited caching capacity case $KL\ll N$.}

The secrecy energy efficiency $\Omega$ in \eqref{max_see} is impacted by various aspects, including the secrecy throughput for each transmission scheme and the MBS/SBS power. All these elements make the caching allocation sophisticated to devise.
Inspired by the fact that the MBS power $P_m$ is typically much larger than the SBS power $P_s$, we focus on a plausible situation $P_m\ge KP_s$ along with a highly concentrated file popularity $\tau >1$. Nevertheless, the design for general cases can be executed similarly, which might incur a tedious classified discussion and a significant computation burden.
The following theorem provides an explicit solution to problem \eqref{max_see}.
{ \begin{theorem}\label{opt_see_theorem}
With hybrid MPC and LCD caching, the optimal allocation of MPC caching $M_E^*$ that maximizes the secrecy energy efficiency $\Omega$ in \eqref{max_see} is given by
\begin{align}\label{opt_m_see}
M_E^{*} 
=\begin{cases}
L, & \Psi^{\rm DF}\ge \Delta\Psi^{\rm B}\xi(L),\\
0, & \Psi^{\rm DF}\le \Delta\Psi^{\rm B}\xi(0),\\
\lceil M^{\circ} \rceil,& \rm otherwise,\\
\end{cases}
\end{align}
where $\Delta\Psi^{\rm B}\triangleq\Delta P_1\Psi^{\rm FB}+\Delta P_2\Psi^{\rm DB}$ denotes the aggregate throughput gain of the DBF and FOT schemes over the BSR scheme, $\xi(M)$ is an increasing function of $M$ given below:
\begin{equation}\label{xi_m}
\xi(M)= \frac{K_1(M+1)^{\tau}}{\Delta P_1(K_L-K_1M)^{\tau}+\Delta P_2K(L+1)}>0,
\end{equation}
and $M^{\circ}$ is the unique root of the equation $\xi(M)=\frac{\Psi^{\rm DF}}{\Delta\Psi^{\rm B}}$.
\end{theorem}}
\begin{proof}
	Please refer to Appendix \ref{appendix_opt_see_theorem}.
\end{proof}

Theorem \ref{opt_see_theorem} provides quite a few interesting insights into the optimal caching allocation:

1) $M_E^* = L$ always holds if $\Delta\Psi^{\rm B}\le 0$. 
The condition $\Delta\Psi^{\rm B}\le 0$ is equivalent to 
\begin{equation}
\frac{P_m+P_s}{KP_s}\le\frac{\Psi^{\rm DF}+\Psi^{\rm B}\left[1-(N+1)^{1-\tau}\right]}{\Psi^{\rm DF}+\Psi^{\rm F}\left[1-(N+1)^{1-\tau}\right]}.
\end{equation}
 This indicates that, when the ratio $(P_m+P_s)/(KP_s)$ (or $\Psi^{\rm B}/\Psi^{\rm F}$) of the BSR scheme over the FOT scheme is lower (or larger) than a certain threshold, it is more conductive to employ MPC caching. This is because, compared to the FOT scheme, in this circumstance the BSR scheme can contribute significantly to improving the secrecy energy efficiency, since it achieves a favorable throughput while not consuming too much power. Hence, we are supposed to increase the probability of adopting the BSR scheme, i.e., activating the MPC caching mode alone.
  
 2) When the BSR can only provide a slight or even no benefit for the secrecy energy efficiency, e.g., an overly high MBS power or low secrecy throughput, the optimal caching allocation is dominated by the throughput gap $\Psi^{\rm DF}$ between the DBF and FOT schemes. To be more precise: i) MPC caching is more rewarding for a significant throughput gap $\Psi^{\rm DF}\ge \Delta\Psi^{\rm B}\xi(L)$; ii) LCD caching becomes preferred for a marginal throughput gap $\Psi^{\rm DF}\le \Delta\Psi^{\rm B}\xi(0)$; iii) for a moderate throughput gap, hybrid caching is necessary for guaranteeing a high level of secrecy energy efficiency. The fundamental reason is that, through adjusting the allocation between MPC and LCD caching, we can strike a superior balance between throughput and power consumption.
 
 3) Although the optimal $M^{\circ}$ does not appear in an explicit form in \eqref{opt_m_see}, we can prove that 
 \begin{equation}
 d M^{\circ}/d\tau=-\frac{\partial Y(M^{\circ})/\partial \tau}{\partial Y(M^{\circ})/\partial M^{\circ}}<0,
 \end{equation}
 by invoking the derivative rule for implicit functions \cite{Zheng2015Multi} with the equation $Y(M^{\circ})=0$, where $Y(M)$ is defined in Appendix \ref{appendix_opt_see_theorem}.
 This suggests that, we should decrease the allocation of MPC caching as the file popularity becomes increasingly concentrated (i.e., a larger $\tau$), which is as opposed to the growth trend observed in Fig. \ref{OPT_MT}. 
 
\begin{figure}[!t]
	\centering
	\includegraphics[width = 3.0in]{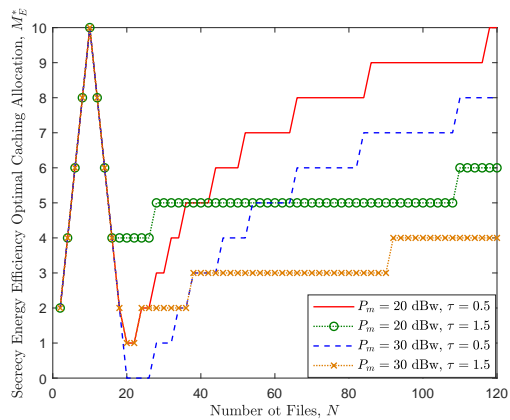}
	\caption{$M_E^*$ vs. $N$ for different values of $P_m$ and $\tau$, with $K=2$, $P_s = 10$ dBw, $\lambda_e=0.01$, $\epsilon=0.2$, and $L = 10$.  }
	\label{OPT_ME}
\end{figure}
{Fig. \ref{OPT_ME} presents the optimal allocation of MPC caching $M_E^*$ versus the number of files $N$. 
We observe that for the large caching capacity case $KL\ge N$, $M_E^*$ first increases linearly and then decreases linearly with $N$, which is similar to Fig. \ref{OPT_MT}. What is different is that $M_E^*$ increases continuously with $N$, since a larger portion of cache resources should be devoted for MPC caching  to maintain a high secrecy throughput. 
We also show that $M_E^*$ decreases with the MBS power $P_m$. 
The underlying reason is that, to counter-balance the increasing backhaul power consumption, a larger portion of cache resources should be reserved for LCD caching to increase the cache hit probability such that the negative impact of backhaul can be mitigated. This manifests the necessity of caching for energy-efficient wireless networks when taking into consideration the backhaul power consumption and delay. Interestingly, it is found that $M_E^*$ increases with the file popularity skewness $\tau$ at the small $N$ regime whereas decreases with $\tau$ at the large $N$ regime. 
This can be explained as follow: for a small $N$, as the file popularity becomes more concentrated (i.e., a larger $\tau$), the advantage from caching more different files is limited and MPC caching is more favorable to increase the secrecy throughput. However, when $N$ is sufficiently large, the benefit from MPC caching in throughput improvement is outweighed by the growing power consumption, whereas increasing the allocation of LCD caching can lower the consumed power and is thus considerably more energy efficient.}

Fig. \ref{MAX_SEE} compares our proposed hybrid caching policy with two conventional approaches solely employing MPC and LCD caching, respectively. Clearly, hybrid caching can always provide the highest secrecy energy efficiency $\Omega^*$, since it can balance secrecy throughput well with power consumption.  
We show that $\Omega^*$ initially increases and then decreases with the SBS power $P_s$. This is because, a too small $P_s$ leads to a low overall secrecy throughput whereas a large $P_s$ results in a high consumed power; both aspects will impair the secrecy energy efficiency.
Additionally, it is as expected that increasing either the number of SBSs $K$ or each SBS's caching capacity $L$ is beneficial for enhancing the secrecy energy efficiency.
From the last two subgraphs with equal total caching capacity, i.e., $KL = 30$, it is surprise to find that equipping a large-capacity cache unit may be more appealing than deploying more SBSs. We can attribute this phenomenon to two causes: a larger $L$ is more rewarding for increasing the cache hit probability, whereas a larger $K$ results in a higher power consumption. This highlights the superiority of wireless caching in improving energy efficiency for cellular networks. 

\begin{figure}[!t]
	\centering
	\includegraphics[width = 2.97in]{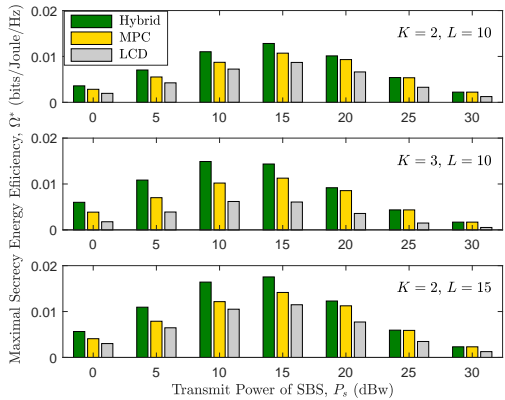}
	\caption{$\Omega^*$ vs. $P_s$ for different values of $K$ and $L$, with $P_m = 30$ dBw,  $\lambda_e=0.01$,  $\epsilon=0.3$, $N = 100$, and $\tau = 1.5$.}
	\label{MAX_SEE}
\end{figure}

\section{Conclusions}
In this paper, we explore the potential of physical-layer security in cache-enabled heterogeneous cellular networks coexisting with PPP distributed eavesdroppers. 
We analyze and optimize the secrecy throughput and secrecy energy efficiency under a hybrid MPC and LCD caching policy along with cooperative transmissions between the MBS and SBSs. 
We provide various interesting insights into the behavior of the optimal transmission rates and the allocation of hybrid caching. We reveal that hybrid caching can always outperform the exclusive use of either MPC or LCD caching in terms of both secrecy throughput and energy efficiency performance. 

\appendix
\subsection{Proof of Theorem \ref{theorem_pco_dbf}}
\label{appendix_theorem_pc_db}
Recall the COP $\mathcal{O}_{co}$ given in \eqref{pco_dbf} and let $x_k =  |h_{s_k,o}|r_{s_k,o}^{-\alpha/2}$, then $\mathcal{O}_{co}$ can be rewritten as 
\begin{align}\label{pc_db_app}
\mathcal{O}_{co}
&=\mathbb{P}\left\{\sum_{k=1}^K x_k<\sqrt{\beta_t/P_s}\right\}.
\end{align}
Note that $x_k$ obeys the Rayleigh distribution with the probability density function (PDF)  $f_{x_k}(x_k)=2{r_{s_k,o}^{\alpha}}{x_k}e^{-r_{s_k,o}^{\alpha}{x_k^2}}$. Due to the independence among $\{x_k\}_{k=1}^K$, the joint PDF can be given by 
\begin{equation}\label{joint_pdf}
f_{x_1,\cdots,x_K}(x_1,\cdots,x_K) =\prod_{k=1}^K 2{r_{s_k,o}^{\alpha}}{x_k}e^{-{r_{s_k,o}^{\alpha}x_k^2}}.
\end{equation}
Substituting \eqref{joint_pdf} into \eqref{pc_db_app} and 
changing the variable $x_k\rightarrow \sqrt{\beta_s/P_s}y_k$ arrives at \eqref{pco_dbf_exact}.

\subsection{Proof of Theorem \ref{theorem_opt_m2}}\label{appendix_theorem_opt_m2}
Treat $M$ as a continuous variable, and we can compute the derivative $\frac{d\bar\Psi}{dM}$ from \eqref{max_st}, 
\begin{equation}\label{dt_cm}
\frac{d\bar\Psi}{dM} 
= \frac{(\tau-1)(M+1)^{-\tau}}{1-(N+1)^{1-\tau}}G(M),
\end{equation}
where $G(M) =\Psi^{\rm DF}-\Psi^{\rm FB}K_1\left(\frac{M+1}{K_L-K_1M}\right)^{\tau} $.
The sign of $\frac{d\bar\Psi}{dM}$ is consistent with that of $G(M)$ and is closely related to the sign of $\Psi^{\rm FB}$. 
Since $\Psi^{\rm DF}>0$, if $\Psi^{\rm FB}\le 0$, we have $G(M)>0$. Hence, $\bar{\Psi}$ increases with $M$ and reaches the maximum at $M=L$.
For the situation  $\Psi^{\rm FB}>0$ such that $G(M)$ decreases with $M$, we distinguish the following three cases.

 Case 1: If $G(L)>0$, i.e., ${\Psi^{\rm DF}}>(K-1){\Psi^{\rm FB}}$, $\frac{d\bar\Psi}{dM}$ is positive, which indicates that $\bar\Psi$ is an increasing function of $M$. Hence, the maximal $\bar\Psi$ is achieved at $M=L$. 
 
Case 2: If $G(0)<0$, i.e., ${\Psi^{\rm DF}}<\frac{K-1}{(KL+1)^{\tau}}{\Psi^{\rm FB}}$, $\frac{d\bar\Psi}{dM}$ is negative, which indicates that $\bar\Psi$ is a decreasing function of $M$. Hence, the maximal $\bar\Psi$ is achieved at $M=0$. 

Case 3: If $G(L)\le 0\le G(0)$, i.e., $\frac{K-1}{\left(KL+1\right)^{\tau}}\le\frac{\Psi^{\rm DF}}{\Psi^{\rm FB}}\le K-1$, $\frac{d\bar\Psi}{dM} $ is first positive and then negative as $M$ increases, i.e., $\bar\Psi$ first increases and then decreases with $M$.
Hence, the maximal $\bar\Psi$ is achieved at the zero-crossing of $\frac{d\bar\Psi}{dM} $. Solving the equation $\frac{d\bar\Psi}{dM^*}=0 $ completes the proof.

\subsection{Proof of Theorem \ref{opt_see_theorem}}\label{appendix_opt_see_theorem}
We start by calculating the derivative $\frac{d\Omega}{dM}$ from \eqref{max_see},
\begin{equation}
\frac{d\Omega}{dM} = \frac{(\tau-1)(M+1)^{-\tau}(K_L-K_1M)^{-\tau}}{\left[\Delta P_1+\Delta P_2(K_L\!-\!K_1M)^{1-\tau}\right]^2}Y(M),
\end{equation}
where $Y(M)=I(M)\Psi^{\rm DF}-K_1(M+1)^{\tau}\Delta \Psi^{\rm B}$, with $I(M)=\Delta P_1(K_L-K_1M)^{\tau}+\Delta P_2K(L+1)>0$ and $\Delta \Psi^{\rm B}$ defined in \eqref{opt_m_see}.
Considering $N\gg 1$ and $\tau>1$, we have $(1+N)^{1-\tau}=0$ such that $\Delta P_1>0$ and $I(M)$ decreases with $M$.
Since $\Psi^{\rm DF}>0$, if $\Delta \Psi^{\rm B}\le0$, we have $Y(M)>0$. Hence, $\Omega$ increases with $M$ and reaches the maximum at $M=L$.
For the case $\Delta \Psi^{\rm B}>0$, $Y(M)$ is a decreasing function of $M$. Then, the proof can be completed following Appendix \ref{appendix_theorem_opt_m2}.

\end{document}